\documentclass[11pt]{article}
\usepackage{fullpage}
\usepackage{epsfig}
\usepackage{amssymb}
\usepackage{amsmath}
\usepackage{pifont}

\makeatletter \@beginparpenalty=10000 \makeatother

\mathcode`\0="0030      %
\mathcode`\1="0031
\mathcode`\2="0032
\mathcode`\3="0033
\mathcode`\4="0034
\mathcode`\5="0035
\mathcode`\6="0036
\mathcode`\7="0037
\mathcode`\8="0038
\mathcode`\9="0039

\makeatletter%
\newcommand{\singlespacing}{\let\CS=
\@currsize\renewcommand{\baselinestretch}{1}\tiny\CS}
\newcommand{\singlespacingplus}{\let\CS=
\@currsize\renewcommand{\baselinestretch}{1.25}\tiny\CS}
\newcommand{\doublespacing}{\let\CS=
\@currsize\renewcommand{\baselinestretch}{1.75}\tiny\CS}
\newcommand{\extradoublespacing}{\let\CS=
\@currsize\renewcommand{\baselinestretch}{1.9}\tiny\CS}
\newcommand{\draftspacing}{\let\CS=
\@currsize\renewcommand{\baselinestretch}{2.0}\tiny\CS}
\newcommand{\hugedraftspacing}{\let\CS=
\@currsize\renewcommand{\baselinestretch}{2.4}\tiny\CS}
\makeatother%

\newcommand{\normalspacing}{\singlespacing}

\newcommand{\calc}{\ensuremath{{\cal C}}}
\newcommand{\cmp}{\ensuremath{{\text{\rm cmp}_{f,g}}}}
\newcommand{\cmph}{\ensuremath{{\text{\rm cmp}_{h,h}}}}
\newcommand{\powercompare}{\ensuremath{{\text{\rm PowerCompare}}}}
\newcommand{\compare}{\ensuremath{{\text{\rm Compare}}}}
\newcommand{\shapley}{\mathrm{SS}}
\newcommand{\banzhaf}{\mathrm{Banzhaf}}
\newcommand{\rawbanzhaf}{{\banzhaf^*}}
\newcommand{\rawshapley}{{\shapley^*}}
\newcommand{\success}{\ensuremath{{\text{\rm succ}}}}

\newcommand{\acc}{\ensuremath{{\text{\rm \#acc}}}}
\newcommand{\naturals}{\mathbb{N}}

\newcommand{\ph}{\ensuremath{{\rm PH}}}

\newcommand{\p}{{\rm P}}

\newcommand{\calS}{\ensuremath{{\cal S}}}
\newcommand{\np}{{\rm NP}}

\newcommand{\subsetsum}{{\rm SubsetSum}}
\newcommand{\xthreec}{{\rm X3C}}

\newcommand{\sharpp}{{\rm \#P}}

\newcommand{\pp}{\ensuremath{{\rm PP}}}

\newtheorem{theorem}{Theorem}[section]

\newtheorem{definition}[theorem]{Definition}

\newtheorem{lemma}[theorem]{Lemma}
\newtheorem{fact}[theorem]{Fact}
\newenvironment{proof}{\sproof}{\bigskip}

\newcommand{\sproof}{\noindent{\bf Proof.}\quad}

\newcommand{\qedblob}{\mbox{\rule[-1.5pt]{5pt}{10.5pt}}}
\renewcommand\qedblob{\mbox{\ding{113}}}
\def\literalqed{{\ \nolinebreak\hfill\mbox{\qedblob\quad}}}

\def\qed{\literalqed}

\newcommand{\selfnote}[1]{}

\begin{document}
\title{The Complexity of Power-Index Comparison\thanks{Also appears as URCS TR-2008-929.}}
\author{Piotr Faliszewski\thanks{ Work done in part while visiting
Heinrich-Heine-Universit\"at D\"usseldorf, Germany.
Supported in part by grant NSF-CCF-0426761.} \\
Department of Computer Science \\ University of Rochester \\ Rochester, NY 14627
\and 
Lane A. Hemaspaandra\thanks{Supported in part by grant NSF-CCF-0426761, 
a TransCoop grant, and a Friedrich Wilhelm Bessel Research Award.} \\
Department of Computer Science \\ University of Rochester \\ Rochester, NY 14627
}

\date{January 29, 2008}

\maketitle              %

\begin{abstract}
  We study the complexity of the following problem: Given two weighted
  voting games $G'$ and $G''$ that each contain a player $p$, in which
  of these 
games
is $p$'s power index 
value
higher? We study this problem
  with respect to both the Shapley-Shubik
power index~\cite{sha-shu:j:powerindex} 
and
  the Banzhaf power index~\cite{ban:j:banzhaf,dub-sha:b:polsci:banzhaf}. 
Our main result is that
  for both of these power indices the problem is complete for 
  probabilistic polynomial time (i.e., is $\pp$-complete).  We
  apply our results to partially resolve some recently proposed 
  problems regarding the complexity of weighted voting games.
  We also study the complexity of the raw Shapley-Shubik
  power index.  
  Deng and
  Papadimitriou~\cite{den-pap:j:solution-concepts} showed that
  the raw Shapley-Shubik power index is
  $\sharpp$-metric-complete.  We strengthen this
  by showing that the raw Shapley-Shubik power index is
  many-one complete for $\sharpp$.
  And our strengthening cannot possibly 
  be further improved to parsimonious completeness, since we observe that,
  in contrast with the raw Banzhaf power index, the raw Shapley-Shubik
  power index is not $\sharpp$-parsimonious-complete.
\end{abstract}

\normalspacing

\section{Introduction}
In an
abstract, direct
democracy, each member in a certain sense 
has equal potential for impact on the decisions that
the society makes. However, in many practical decision-making
scenarios it is reasonable to give up this noble idea and consider
weighted voting instead. 
Here are a few motivating examples.  In a country divided into
districts it makes sense to give each district voting power
proportional to its population (consider, e.g., 
the US House of Representatives
or various decision making processes within the European Union).
In fact, the power that various apportionment methods give 
to the US states in its House of Representatives has 
been studied in terms of how well it is proportional 
to the sizes of the states~\cite{hem-raj-set-zim:j:power-indices}.
In a business
setting, stockholders in a company might hope to have voting power
proportional to the amount of stock they own.
Within computer science,
Dwork et al.~\cite{dwo-kum-nao-siv:c:rank-aggregation} suggested
building a meta search engine for the web via treating other search
engines as voters in an election.  It would only be natural to weigh
the participating search engines with their (quantified in some way)
quality. Naturally, one can provide many other examples.

The focus of this paper is on the computational complexity of the
following issue: Given an individual and two weighted voting scenarios
(in each of them our individual might have different weight and each
scenario might involve different sets of voters with different
weights), in which one of them is our individual more influential? (We
provide a formal definition of this problem in Section~\ref{sect:prelim:votinggames}.)
This problem has a very natural motivation. For example, consider a
company that wishes to join some business consortium and has a choice
among several consortia (e.g., consider an airline deciding which airline
alliance to join). It is natural to assume that within each consortium
companies make decisions via weighted voting, with companies weighted,
e.g., via their size or revenue or some combination thereof. In a
political context, members of the European Union sometimes try to promote
new schemes of distributing vote weights among EU members.  It is
important for the countries involved to see which scheme is better for
them. One can easily give many other applications of the issue we
study.

Formally, we model the above problem via comparing the values of power
index functions---in our case those of Shapley and
Shubik~\cite{sha-shu:j:powerindex} and of 
Banzhaf (\cite{ban:j:banzhaf}, see also~\cite{dub-sha:b:polsci:banzhaf})---of 
a particular player within two given weighted voting games.  Our
main result is that this problem is $\pp$-complete for
both the Shapley-Shubik power index
and the Banzhaf power index.  Let us now define our problem formally.

\subsection{The Power-Index Comparison Problem}
\label{sect:prelim:votinggames}

We model weighted voting via so-called weighted voting games.  An
$n$-player weighted voting game is a sequence of $n$ nonnegative
integer weights, $w_1, \ldots, w_n$, together with a quota $q$. 
We denote it as $(w_1, \ldots, w_n; q)$. We refer to the player with
weight $w_i$ as the $i$'th player.  Weighted voting games model the
following scenario: The players are given a yes/no question (e.g.,
should we lower the taxes? should we buy out our competitors?) and
each player either agrees (answers \emph{yes}) or disagrees (answers
\emph{no}). If the total weight of the voters who agree is at least as
high as the quota then the result of the game is \emph{yes} and
otherwise it is \emph{no}.

Let $G$ be a voting game $(w_1, \ldots, w_n; q)$. Any subset of $\{1,
\ldots, n\}$ is a coalition in $G$. We say that a coalition $S$ is
successful if $\sum_{i\in S} w_i \geq q$. We define $\success_G(S)$ to
be $1$ if $S$ is a successful coalition for $G$ and to be $0$ otherwise.

Interestingly, the relation between the effective power of a player
within a voting game and his or her weight is not as simple as one
might think. Consider game $G = (8,7,2; 9)$, i.e., a game with quota
$q = 9$ and three players with weights $8$, $7$, and $2$,
respectively. It is easy to see that in this game any coalition of at
least two players is successful. In effect, each of the players
can influence the final result of the game to exactly the same degree,
regardless of the fact that their weights 
differ significantly.
Thus when analyzing weighted voting games it is standard to measure
players' power using, e.g., the Shapley-Shubik 
power index~\cite{sha-shu:j:powerindex}
or the Banzhaf power index~\cite{ban:j:banzhaf,dub-sha:b:polsci:banzhaf}.

In essence, these power indices measure the probability that, assuming
some coalition formation model, our designated player is critical for
the forming coalition. By \emph{critical} we mean here that the
coalition is successful with our designated player but is not
successful without him or her.

Let $G = (w_1, \ldots, w_n; q)$ be a voting game, let $i$ be a player
in this game, and let $N = \{1, \ldots, n\}$ be the set of all players
of $G$. The value of the Banzhaf power index of $i$ in $G$
is defined as $\banzhaf(G,i) =
\frac{\rawbanzhaf(G,i)}{2^{n-1}}$, where $\rawbanzhaf(G,i)$ is the
\emph{raw} version of the index,
\[
\rawbanzhaf(G,i) = \sum_{S \subseteq N - \{i\}}(\success_G(S \cup \{i\}) - \success_G(S)).\]
The Shapley-Shubik power index of player $i$ in game $G$
is defined as $\shapley(G,i) = \frac{\rawshapley(G,i)}{n!}$, where
$\rawshapley(G,i)$ is the raw version of the index,
\[\rawshapley(G,i) = \sum_{S \subseteq N - \{i\}}
\|S\|!(n-\|S\|-1)!(\success_G(S \cup \{i\}) - \success_G(S)).\]

Intuitively, $\banzhaf(G,i)$ gives the probability that a randomly
chosen coalition of players in $N - \{i\}$ is not successful but would
become successful had player $i$ joined in.  The intuition for
the Shapley-Shubik index is that we count the proportion of permutations for
which a given player is pivotal. Given a permutation $\pi$ of $\{1,
\ldots, n\}$, the $\pi(i)$'th player is pivotal if it holds that the coalition
$\{\pi(1), \pi(2), \ldots, \pi(i)\}$ is successful and the coalition
$\{\pi(1), \pi(2), \ldots, \pi(i-1)\}$ is not. This permutation-based intuition
is motivated by the view of the successful-coalition formation as the
process of players joining in in random order. Naturally, the first
player that makes the coalition successful is crucial and so the idea
is to measure power via counting how often our player-of-interest
   is pivotal.

The focus of this paper is on the computational complexity analysis of
the following problem.

\begin{definition}\label{def:powercompare}
Let $f$ be either the Shapley-Shubik or the Banzhaf power index.  By
$\powercompare_f$ we mean the problem where the input $(G',G'',i)$
contains two weighted
voting games, $G' = (w'_1, \ldots, w'_n, q')$ and $G'' = (w''_1,
\ldots, w''_n, q'')$, and an integer $i$, $1 \leq i \leq n$, and where
we ask whether $f(G',i) > f(G'',i)$.
\end{definition}

Note that in the above definition we assume that both games have the
same number of players. At first this might seem to be a weakness 
but it is easy to see that given two games 
with different numbers of players
we can easily pad the smaller one with
weight-$0$ players.  On the other hand, the assumption that both games
have the same number of players allows us to solve the problem via
comparing the raw values of the index: The scaling factor for both
games is the same and thus it does not affect the result of the
comparison.

\subsection{Computational Complexity}
We briefly review some 
notions and notations.

We fix the alphabet $\Sigma = \{0,1\}$, and we assume that all the
problems we consider are encoded in a natural, efficient manner over $\Sigma$.  By
$|\cdot|$ we mean the length function. We assume $\langle\cdot, \cdot
\rangle$ to be a standard, natural pairing function such that
$|\langle x, y \rangle| = 2(|x| + |y|)+2$.

The main result of this paper, Theorem~\ref{thm:main},
says that
the power-index comparison problem
is $\pp$-complete both for
the Shapley-Shubik 
power index and for the Banzhaf power index.
The class
$\pp$, probabilistic polynomial time, 
was defined by Simon~\cite{sim:thesis:complexity} and
Gill~\cite{gil:j:prob-tms}.  A language $L \subseteq \Sigma^*$ belongs
to $\pp$ if and only if there exists a polynomial $p$ and a
polynomial-time computable relation $R$ such that
$x \in L \iff \| \{ w \in \Sigma^{p(|x|)} \mid R(x,w)\ \mbox{holds} \}   \| > 2^{p(|x|)-1}$.
PP captures the set of languages having a probabilistic Turing machine 
that on precisely the elements of the set has strictly more than 50\% probability
of acceptance.  
Let us mention that $\pp$ is a very powerful class.  For example, it
is well-known that $\np$ 
is a subset of $\pp$
(as are even various larger classes). 
Via Toda's
Theorem~\cite{tod:j:pp-ph}, we know that $\ph \subseteq \p^\pp$. That
is, $\pp$ is at least as powerful as polynomial-time
hierarchy, give or take the flexibility of 
polynomial-time Turing reductions.
Many other properties of $\pp$ have been established in the
literature.

Let us now recall the definition of the class
$\sharpp$~\cite{val:j:permanent}.  For each $\np$ machine $N$ (i.e.,
for each nondeterministic polynomial-time machine $N$), by $\acc_N(x)$
we mean the number of accepting computation paths of $N$ running with
input $x$.  A function $f$, $f: \Sigma^* \rightarrow \naturals$,
belongs to $\sharpp$ if and only if there is an $\np$ machine $N$ such
that $(\forall x \in \Sigma^*)[f(x) = \acc_N(x)]$.  $\sharpp$ is, in
a very loose sense, a functional counterpart of $\pp$. For example, $\p^\sharpp
= \p^\pp$~\cite{bal-boo-sch:j:sparse}.  More typically, $\sharpp$
is described as the counting analogue of $\np$.

As is usual, we say that a language $L$ is hard for a complexity
class $\calc$ if  every language in $\calc$
polynomial-time many-one reduces to $L$. If in addition $L$ belongs
to $\calc$ then we say that $L$ is $\calc$-complete.
A language $A$
polynomial-time many-one reduces to a language $B$ if there exists a
polynomial-time computable function $f$ such that for each string $x
\in \Sigma^*$ it holds that $x \in A \iff f(x) \in B$.  On the other
hand, there is no one agreed-upon notion of completeness for function
classes. For example, Valiant~\cite{val:j:permanent} in his seminal
paper used Turing reductions but other people have preferred notions
such as Krentel's metric reductions~\cite{kre:j:optimization},
Zank{\'{o}}'s many-one reductions (for functions)~\cite{zan:j:sharp-p}, and
Simon's~\cite{sim:thesis:complexity} parsimonious reductions.

In the context of power index functions, Prasad and Kelly~\cite{pra-kel:j:voting}
(implicitly) showed that
the (raw) Banzhaf power index is 
$\sharpp$-parsimonious-complete
and Deng and
Papadimitriou~\cite{den-pap:j:solution-concepts} 
established that the (raw) Shapley-Shubik power index is
$\sharpp$-metric-complete
(regarding the complexity analysis of power
indices, we also mention the paper of Matsui and
Matsui~\cite{mat-mat:j:power-index-complexity}). We now review
parsimonious and metric reductions, as those 
underpin the notions of 
parsimonious-completeness and metric-completeness.

\begin{definition}
\begin{enumerate}
\item \cite{kre:j:optimization}
  A function $f: \Sigma^* \rightarrow \naturals$ metric
  reduces to a function $g: \Sigma^* \rightarrow \naturals$ if there
  exist two polynomial-time computable functions, $\varphi$ and
  $\psi$, such that $(\forall x \in \Sigma^*)[f(x) =
  \psi(x,g(\varphi(x)))].$  

\item \cite{zan:j:sharp-p}
  A function $f: \Sigma^* \rightarrow \naturals$ many-one
  reduces to a function $g: \Sigma^* \rightarrow \naturals$ if there
  exists two  polynomial-time computable functions, $\varphi$ and
  $\psi$, such that $(\forall x \in \Sigma^*)[f(x) =
  \psi(g(\varphi(x)))].$\footnote{Note that Zank{\'{o}}'s 
many-one reduction is a analogue 
for functions of the standard many-one reduction notion 
for sets.  
To avoid confusion, we mention to the reader 
that 
the term 
``functional many-one reduction'' (which we do not use here) 
is sometimes 
used in the literature~\cite{vol:c:function} 
as a synonym 
for ``parsimonious reductions.''
}
\item \cite{sim:thesis:complexity}
  $f$ parsimoniously reduces to $g$ if there is a
  polynomial-time computable function $\varphi$ such that $(\forall x
  \in \Sigma^*)[f(x) = g(\varphi(x))]$.
\end{enumerate}
\end{definition}

Note that (a)~if $f$ parsimoniously reduces to $g$, then 
$f$ 
many-one reduces to $g$, 
and (b)~if
$f$ 
many-one reduces to $g$, 
then $f$ metric reduces to $g$.  Given a
function class $\calc$, we say that a function $f$ is
$\calc$-parsimonious-complete if $f \in \calc$ and each function in $\calc$
parsimonious reduces to $f$. $\calc$-metric-completeness and
$\calc$-many-one-completeness are defined analogously. Typically,
parsimonious-complete functions are easier to work with than functions
that are merely metric-complete or many-one-complete. In particular,
our proof of Theorem~\ref{thm:shapley} is more involved than our proof
of Theorem~\ref{thm:banzhaf} because, as we note, the raw
Shapley-Shubik power index is not parsimoniously complete.

\section{Main Results}

Our main result, Theorem~\ref{thm:main}, says that the power index
comparison problem is $\pp$-complete. This section is devoted to
building the infrastructure for Theorem~\ref{thm:main}'s proof and
giving that proof. We also show that the raw Shapley-Shubik power
index is $\sharpp$-many-one-complete but not
$\sharpp$-parsimonious-complete.

\begin{theorem}
  \label{thm:main}
  Let $f$ be either the Banzhaf or the Shapley-Shubik power index. The problem
  $\powercompare_f$ is $\pp$-complete.
\end{theorem}

We start via showing $\pp$-membership of a problem closely related to
our $\powercompare_\banzhaf$ and $\powercompare_\shapley$ problems.  
Let $f$ be
a $\sharpp$ function and let $\compare_f$ be the language $\{ \langle
x, y \rangle \mid x,y \in \Sigma^* \land f(x) > f(y) \}$.  
($\powercompare_\banzhaf$ and $\powercompare_\shapley$ are
essentially, up to a minor definitional issue, incarnations of
$\compare_f$ for appropriate functions $f$.)

\begin{lemma}\label{thm:pp-member}
  Let $f$ be a $\sharpp$ function. The language
  $\compare_f$ %
  is in $\pp$.
\end{lemma}
\begin{proof}
Let $f$ be an arbitrary $\sharpp$ function and let $N$ be an
$\np$ machine such that $f = \acc_N$.  Without the loss of generality,
we assume that 
 there is a polynomial $q$ such that for each input $x
\in \Sigma^*$ all computation paths of $N$ make exactly
$q(|x|)$ binary nondeterministic choices. 
Thus each computation path of $N$ on input $x$
can be represented as a string $w$ in $\Sigma^{q(|x|)}$.

In order to show that $\compare_f$ is in $\pp$ we need to provide a
polynomial-time computable relation $R$ and a polynomial $p$ such that
for each string $z = \langle x, y \rangle$ it holds that:
$z \in \compare_f \iff \| \{ w \in \Sigma^{p(|z|)} \mid R(z,w)\
\mbox{holds} \} \| > 2^{p(|z|)-1}$.
We now define such $R$ and $p$. Let us fix two strings, $x$ and $y$,
and let $z = \langle x,y\rangle$ and $n = |z|$. We define $p(n) =
q(n)+1$ and, for each string $w = w_0w_1\ldots w_{p(n)-1} \in
\Sigma^{p(n)}$, we define $R(z,w)$ as follows:
\begin{description}
\item [Case 1.] If $w_0 = 0$ then $R(z, w)$ is \emph{true} exactly if
  the string $w_1, \ldots, w_{q(|x|)}$ denotes an accepting computation
  path of $N$ on $x$
  and the symbols $w_{q(|x|)+1}$ through $w_{p(n)-1}$ are all $0$. $R(z, w)$ is false otherwise.

\item [Case 2.] If $w_0 = 1$ then $R(z, w)$ is \emph{false} exactly if
  the string $w_1, \ldots, w_{q(|y|)}$ denotes an accepting computation
  path of $N$ on $y$
  and the symbols $w_{q(|x|)+1}$ through $w_{p(n)-1}$ are all $0$. $R(z, w)$ is true otherwise.
\end{description}
Via analyzing the above two cases it is easy to see that there are
exactly $f(x) + (2^{p(n)-1} - f(y)) = f(x) - f(y) +2^{p(n)-1}$
strings $w \in \Sigma^{p(n)}$ for which $R(z, w)$
is true.  This value is greater than $2^{p(n)-1}$ if and only if
$f(x) > f(y)$.
Thus the relation $R$ and the polynomial $p$ jointly
witness that $\compare_f$ belongs to $\pp$.~\qed
\end{proof}

Lemma~\ref{thm:pp-member} gives an upper bound on the complexity of
$\compare_f$ (assuming that $f \in \sharpp$). We now prove a
matching lower bound, $\pp$-completeness, for the case that $f$ is
$\sharpp$-parsimonious-complete.

\begin{lemma}\label{thm:pp-com}
  Let $f$ be a $\sharpp$-parsimonious-complete function. The language
  $\compare_f$ is $\pp$-complete.
\end{lemma}

\begin{proof}
  Let $f$ be a $\sharpp$-parsimonious-complete function. Via
  Lemma~\ref{thm:pp-member} we know that $\compare_f$ is in $\pp$ and
  thus to show $\pp$-completeness it remains to show $\pp$-hardness.
  We do so via reducing an arbitrary $\pp$ language $L$ to $\compare_f$.

  Let $L$ be an arbitrary $\pp$ language.  By definition, there exists
  a polynomial-time relation $R$ and a polynomial $p$ such that for
  each string $x \in \Sigma^*$ it holds that $x \in L \iff \| \{ y \in
  \Sigma^{p(|x|)} \mid R(x,y)\ \mbox{holds} \} \| > 2^{p(|x|)-1}.$ 
  We
  define two functions, $g_1$ and $g_2$, such that $g_1(x) = \| \{ y
  \in \Sigma^{p(|x|)} \mid R(x,y)\ \mbox{holds} \} \|$ and $g_2(x) =
  2^{p(|x|)-1}$. It is easy to see that both $g_1$ and $g_2$ are in
  $\sharpp$.  $g_1$ can be computed via a an $\np$ machine that on input
  $x$ guesses a binary string $y$ of length $p(|x|)$ and accepts if
  and only if $R(x,y)$ holds. $g_2$ can be computed via a machine that
  on input $x$ guesses a binary string of length $2^{p(|x|)-1}$ and
  then accepts. Naturally, $x \in L$ if and only if $g_1(x) > g_2(x)$.

  Since $f$ is $\sharpp$-parsimonious-complete, both $g_1$ and $g_2$
  parsimoniously reduce to $f$. Let $\varphi_1$ be the reduction function
  for $g_1$ and let $\varphi_2$ be the reduction function for $g_2$.
  We have that for each string $x$ it holds that $g_1(x) =
  f(\varphi_1(x))$ and $g_2(x) = f(\varphi_2(x))$.

  Our reduction from $L$ to $\compare_f$ works as follows. On input
  $x$ we output the string $z = \langle \varphi_1(x),
  \varphi_2(x) \rangle$. Clearly, this can be done in polynomial time. 
  To show correctness it is enough to recall that $x \in L$ if and
  only if $g_1(x) > g_2(x)$, which is equivalent to testing whether $z$ is
  in $\compare_f$.
  Since $L$ was chosen as an arbitrary $\pp$ language, this proves $\pp$-completeness.~\qed
\end{proof}

We are almost ready to show that $\powercompare_\banzhaf$ is
$\pp$-complete. However, in order to do so, we need to justify 
the claim that
the raw version of the Banzhaf power index is
$\sharpp$-parsimonious-complete. (This was shown implicitly in the
work of Prasad and Kelly~\cite{pra-kel:j:voting}, but we feel that it
is important to explicitly outline the proof.)

One of our important tools here (and later on) is the function
$\#\xthreec$.  The input to the $\xthreec$ problem is a set $B = \{b_1,
\ldots, b_{3k}\}$ and a family ${\mathcal S} = \{S_1, \ldots,
S_n\}$ of 3-element subsets of $B$. The $\xthreec$ problem asks whether
there exists a collection of exactly $k$ sets in ${\mathcal S}$ whose
union is $B$. $\#\xthreec(B,{\mathcal S})$ 
is the number of solutions of the $\xthreec$ instance $(B,{\mathcal S})$.

Hunt et al.~\cite{hun-mar-rad-ste-:j:planar-counting-problems} showed
that $\#\xthreec$ is parsimonious complete for $\sharpp$. This is very
useful for us as the standard reduction from $\#\xthreec$ to
$\#\subsetsum$ (see, e.g., \cite[Theorem~9.10]{pap:b:complexity};
$\#\subsetsum$ is the function that accepts as input a
vector of nonnegative integers $(s_1, \ldots, s_n; q)$
and returns
the number of subsets of $\{s_1, \ldots, s_n\}$ that sum up to $q$) is
parsimonious and Prasad and Kelly's reduction from $\#\subsetsum$ to
$\rawbanzhaf$ (the raw version of Banzhaf's power index) is parsimonious
as well. Since $\rawbanzhaf$ is in $\sharpp$, $\rawbanzhaf$ is
$\sharpp$-parsimonious-complete. Thus the following theorem is,
essentially, a direct consequence of Lemma~\ref{thm:pp-com}.

\begin{theorem}\label{thm:banzhaf}
  $\powercompare_\banzhaf$ is $\pp$-complete.
\end{theorem}
\begin{proof}
  The raw version of the Banzhaf power index is
  $\sharpp$-parsimonious-complete and so, via Lemma~\ref{thm:pp-com},
  $\compare_\rawbanzhaf$ is $\pp$-complete.  Via a slight misuse of
  notation, we can say that $\compare_\rawbanzhaf$ accepts as input
  two weighted voting games, $G'$ and $G''$, and two players, $p'$ and
  $p''$, such that $p'$ participates in $G'$ and $p''$ participates in
  $G''$ and accepts if and only if $\rawbanzhaf(G',p') >
  \rawbanzhaf(G',p'')$. We give a reduction from
  $\compare_\rawbanzhaf$ to $\powercompare_\banzhaf$.

  Let $G',p'$ and $G'',p''$ be our input to the
  $\compare_\rawbanzhaf$ problem.  We can assume
  that $G'$ and $G''$ have the same number of players.  If $G'$ and
  $G''$ do not have the same number of players then it is easy to see
  that the game with fewer players can be padded with players whose
  weight is equal to this game's quota value. 
  Such a padding leaves the raw Banzhaf power index values
  of the game's original players unchanged. (The reason for this is that
  any coalition that includes any of the padding candidates is already 
  winning and so none of the original player's is critical to
  the success of the coalition, and so the coalition 
  does not contribute to original players' power index
  values.)

  We form two games, $K'$ and $K''$, that are identical to games
  $G'$ and $G''$, respectively, except that $K'$ lists player $p'$ as
  first and $G''$ lists player $p''$ as first.
  Our reduction's output is $(K',K'',1)$. 

  Naturally, $\banzhaf(K',1) > \banzhaf(K'',1)$ if and
  only if $\rawbanzhaf(G',p') > \rawbanzhaf(G'',p'')$. Also, clearly,
  $K'$ and $K''$ can be computed in polynomial time. Thus we have
  successfully reduced $\compare_\rawbanzhaf$ to
  $\powercompare_\banzhaf$. This shows $\pp$-hardness of
  $\powercompare_\banzhaf$. $\pp$-membership of
  $\powercompare_\banzhaf$ is, essentially, a simple consequence of
  Lemma~\ref{thm:pp-member}. This completes the proof.~\qed
\end{proof}

Let us now focus on the computational complexity of the power index
comparison problem for the case of Shapley-Shubik. 
It would be nice if
the raw Shapley-Shubik power index 
were $\sharpp$-parsimonious-complete.
If that 
were
the case then we could
establish $\pp$-completeness of $\powercompare_\shapley$ in
essentially the same way as we did for $\powercompare_\banzhaf$.
Thus it is natural to ask whether the Shapley-Shubik power index (i.e., its raw version)
is $\sharpp$-parsimonious-complete.
Prasad and Kelly~\cite{pra-kel:j:voting} 
at the end of their paper, after---in effect---showing
$\sharpp$-parsimonious-completeness of the raw Banzhaf power index
(their Theorem~4),  write:
``Such a straightforward approach does not seem possible with the
Shapley-Shubik [power index].'' We reinforce their intuition by now
proving that the raw Shapley-Shubik power index in fact is \emph{not}
$\sharpp$-parsimonious-complete.

\begin{theorem}\label{thm:shapley-nonparsimonious}
  The raw Shapley-Shubik power index (i.e., $\rawshapley$) is not
  $\sharpp$-parsimonious-complete.
\end{theorem}

\begin{proof}
For the sake of contradiction, let us assume that
$\rawshapley$ is $\sharpp$-parsimonious-complete.  
Thus for each natural number $k$ there is a weighted voting game
$G$ and a player $i$ within $G$ such that $\rawshapley(G,i) = k$. This
is the case because the function $f(x) = x$ belongs to $\sharpp$ (we
assume that the ``output $x$'' is an integer obtained via a standard
bijection between $\Sigma^*$ and $\naturals$) and if $\rawshapley$ is
$\sharpp$-parsimonious-complete then there has to be a parsimonious
reduction from $f$ to $\rawshapley$.

Let $G$ be an arbitrary voting game with $n \geq 4$ players and let
$i$ be a player in $G$. By definition, $\rawshapley(G,i)$ is a sum of
terms of the form $k!(n-k-1)!$, where $k$ is some value in $\{0,
\ldots, n-1\}$. Since $n \geq 4$, each such term is even and thus
$\rawshapley(G,i)$ is even.
The raw Shapley-Shubik
power index of any player in a game with at most $3$ players is at
most $3! = 6$ and  thus there is no input on which
$\rawshapley$ yields the value $7$.  This contradicts the assumption
that $\rawshapley$ is $\sharpp$-parsimonious-complete and completes
the proof.~\qed
\end{proof}

So the well-known result of Deng and
Papadimitriou~\cite{den-pap:j:solution-concepts} that the raw
Shapley-Shubik power index is $\sharpp$-metric-complete cannot be
strengthened to $\sharpp$-parsimonious-completeness.
Theorem~\ref{thm:shapley-nonparsimonious} prevents us from directly
using Lemma~\ref{thm:pp-com} to show that $\powercompare_\shapley$ is
$\pp$-complete. Nonetheless, via the following set of results  not only
do we establish that $\powercompare_\shapley$ is $\pp$-complete, but 
we also
strengthen the result of Deng and Papadimitriou via showing that
the raw Shapley-Shubik power index is $\sharpp$-many-one-complete
(i.e., is $\sharpp$-complete w.r.t.\ Zank{\'{o}}'s 
many-one
reductions~\cite{zan:j:sharp-p}).

To establish our results we need to be able to build $\xthreec$
instances that satisfy certain
properties. Fact~\ref{fact:x3c-transform} below lists three basic
transformations that we use to enforce these properties.

\begin{fact}\label{fact:x3c-transform}
  Let $(B,\calS)$ be an instance of $\xthreec$ and let
  $b_1, b_2, \ldots, b_6$ be elements that do not belong to $B$.
  Let $B_1 = \{b_1, b_2, b_3\}$, $B_2 = \{b_4, b_5, b_6\}$,
  $B_3 = \{b_1, b_4, b_5\}$ and $B_4 = \{b_1, b_4, b_6\}$.
  The following transformations preserve the number of
  solutions of the input instance:
  \begin{enumerate}
  \item $g(B,S) = (B \cup B_1, \calS \cup \{B_1\})$,
  \item $h'(B,S) = (B \cup B_1 \cup B_2, \calS \cup
    \{B_1,B_2,B_3\})$,
  \item $h''(B,S) = (B \cup B_1 \cup B_2, \calS \cup
    \{B_1,B_2,B_3,B_4\})$,
  \end{enumerate}
\end{fact}
In the following lemma we use these transformations to, in some sense,  normalize
$\xthreec$ instances.

\begin{lemma}\label{thm:two-thirds}
  There is a polynomial-time algorithm that given an $\xthreec$
  instance $X = (B,\calS)$ outputs instance $X'' = (B'', \calS'')$ such that
  $\#\xthreec(X'') = \#\xthreec(X)$ and $\frac{\frac{1}{3}\|B''\|}{\|\calS''\|} =
  \frac{2}{3}$.
\end{lemma} 
\begin{proof}
  Let $X = (B,\calS)$ be our input $\xthreec$ instance and let $3k =
  \|B\|$ and $m = \|\calS\|$. Let $g$ and $h''$ be the
  transformations as in Fact~\ref{fact:x3c-transform}. The idea of our
  algorithm is to apply transformation $g$ to $X$ so many times as to
  achieve the $\frac{2}{3}$ ratio.
  Let $t$ be some nonnegative
  integer and let $(B_t,\calS_t) = g^{(t)}(B,\calS)$. We observe that
  $\frac{\frac{1}{3}\|B_t\|}{\|\calS_t\|} = \frac{k+t}{m+t}$ and that if $t = 2m-
  3k$ (assuming this value is nonnegative) then %
  $\frac{k+t}{m+t} = \frac{2}{3}$.

  Our algorithm works as follows. First, we form instance $X' =
  (B',\calS')$ such that $2\|\calS'\| - 3\cdot\frac{1}{3}\|B'\| \geq 0$.  If $2m - 3k
  \geq 0$ then we set $X' = X$ and otherwise we repeatedly apply
  transformation $h''$, until this condition is met.  (It is easy to
  see that $\lceil \frac{3k-2m}{2} \rceil$ applications are
  sufficient.)  Then we derive the instance $X''$ from $X'$ via
  $2\|\calS'\| - 3\cdot\frac{1}{3}\|B'\|$ applications of $g$. That is, $X'' =
  g^{(2\|\calS'\| - 3\cdot\frac{1}{3}\|B'\|)}(X')$.

  Naturally, the algorithm runs in polynomial time. The correctness
  follows via the observation in the first paragraph and the fact that
  transformations $g$ and $h''$ preserve the number of solutions.~\qed
\end{proof}

Finally, we are ready to show that the raw Shapley-Shubik power index
is $\sharpp$-many-one-complete.

\begin{theorem}\label{thm:many-one}
  The raw Shapley-Shubik power index (i.e., $\rawshapley$) is
  $\sharpp$-many-one-complete.
\end{theorem}
\begin{proof}
  The raw Shapley-Shubik power index is in $\sharpp$ and
  thus it remains to show that it is $\sharpp$-many-one-hard. To do
  so, we give a 
many-one reduction from $\#\xthreec'$ to
  $\rawshapley$. $\#\xthreec'$ is a restriction of $\#\xthreec$
  to instances $X = (B,S)$ such that:
  (1) $\frac{\frac{1}{3}\|B\|}{\|\calS\|} = \frac{2}{3}$.
  (2) If $n$ is a nonnegative integer such that $\frac{1}{3}\|B\| = 2n$ and
    $\|\calS\|=3n$ then there is a nonnegative integer $t$ such that
    $n = 4^t$.
  To see that the thus restricted $\#\xthreec$ function is
  $\sharpp$-parsimonious-complete it is enough to consider
  Lemma~\ref{thm:two-thirds} and transformation $h'$ from
  Fact~\ref{fact:x3c-transform}.

  Let $\varphi_s$ be the standard, parsimonious reduction from
  $\#\xthreec$ to $\#\subsetsum$ (see, e.g.,~\cite
  [Theorem~9.10]{pap:b:complexity}). %
  $\varphi_s$ has the property that given an instance $(B, {\cal S})$,
  where $\|B\| = 3k$ and $\|S\| = m$,
  $\varphi_s(B,{\cal S})$ is an instance $(s_1, \ldots, s_m; q)$ of
  $\subsetsum$ such that every subset of $\{s_1, \ldots, s_m\}$ that
  sums up to $q$ has exactly $k$ elements.  Given such an instance
  $(s_1, \ldots, s_m; q)$, Deng and
  Papadimitriou~\cite[Theorem~9]{den-pap:j:solution-concepts} observe
  that the raw Shapley-Shubik power index of the first player in game
  $(1, s_1, \ldots, s_m; q+1)$ is exactly
  $(m - k)!k!\cdot \#\subsetsum(s_1,\ldots,s_n; q)$.  Since
  $\varphi_s$ is parsimonious, this value is equal to $(n - m)!m!\cdot
  \#\xthreec(B,\calS)$.

  We now provide functions $\varphi$ and $\psi$ that constitute a
many-one reduction from $\#\xthreec'$ to
  $\rawshapley$. We need to ensure that for each $\#\xthreec'$
  instance $X$\footnote{We assume that the inputs to $\varphi$ satisfy
    the requirements of being $\#\xthreec'$ instances.  We implicitly
    replace any instance that does not fulfill this requirement with a
    fixed instance that does satisfy it and that has no solutions.} it
  holds that $\#\xthreec'(X) = \psi(\rawshapley(\varphi(X)))$. We
  first describe how to compute $\varphi$ and $\psi$ and then explain
  why they have this property.

  Given $\#\xthreec'$ instance $X$, we compute $\varphi(X)$ as
  follows: We compute $\subsetsum$ instance $\varphi_s(X) = (s_1,
  \ldots, s_n; q)$ and output game $(1, s_1, \ldots, s_n; q+1)$.
  Function $\psi$ is a little more involved. Define $r_1(n) = n!(2n)!$
  and $r_2(n) = n!(2n)!  2^{3n}$.  Given a nonnegative integer $x$, we
  compute $\psi(x)$ using the following algorithm. If $x = 0$ then
  return $0$. Otherwise, find the smallest nonnegative integer $t$
  such that $r_1(4^t) \leq x \leq r_2(4^t)$ and output $\lfloor
  \frac{x}{r_1(4^t)} \rfloor$. If there is no such $t$ then return
  $0$. Function $\psi(x)$ can be computed in polynomial time via
  computing $r_1(4^t)$ and $r_2(4^t)$ for successive values of $t$. It
  is easy to see that we only need to try $O(\log x)$ many $t$'s and
  thus $\psi$ is computable in polynomial time with respect to the
  binary representation of $x$.

  Let us now show that indeed for any $\#\xthreec'$ instance $X$ it
  holds that $\#\xthreec'(X) = \psi(\rawshapley(\varphi(X)))$.  Let $X
  = (B,S)$ be an arbitrary $\#\xthreec'$ instance and let $n$ be a
  nonnegative integer such that $\frac{1}{3}\|B\| = 2n$ and $\|\calS\|
  = 3n$. (The existence of such an $n$ is guaranteed via the fact that in
  any $\#\xthreec'$ instance $\frac{\frac{1}{3}\|B\|}{\|\calS\|} =
  \frac{2}{3}$.)  Via the properties of $\varphi_s$ and $\varphi$ we see
  that
  \[\rawshapley(\varphi(X)) = n!(2n)! \#\xthreec'(X) = r_1(n)\#\xthreec'(X).\]
  It is easy to see that $\#\xthreec'(X) \leq 2^{3n}$ and thus,
  assuming that $\#\xthreec'(X) \geq 1$, we have that
  \[r_1(n) \leq \rawshapley(\varphi(X)) \leq r_2(n).\]
  Via routine calculation we see that for any positive integer $n$ it
  holds that $r_1(4n) > r_2(n)$. Thus the intervals $[r_1(4^t),r_2(4^t)]$
  are disjoint and given $\rawshapley(\varphi(X))$ as input, the function
  $\psi$ correctly identifies the $r_1(n)$ factor and outputs 
  the
  answer $\#\xthreec'(X)$. Clearly, $\psi$ also works correctly when
  $\rawshapley(\varphi(X)) = 0$.~\qed
\end{proof}

\begin{lemma}\label{thm:x3c-transform}
  There is a polynomial-time algorithm that given two $\xthreec$
  instances $X = (B_x,\calS_x)$ and $Y=(B_y,\calS_y)$ outputs two
  $\xthreec$ instances $X'' = (B''_y,\calS''_y)$ and $Y''=(B''_y,\calS''_y)$
  such that
     $\|B''_x\| = \|B''_y\|$,
     $\|\calS''_x\| = \|\calS''_y\|$,
     $\#\xthreec(X) = \#\xthreec(X'')$, and
     $\#\xthreec(Y) = \#\xthreec(Y'')$.
\end{lemma}
\begin{proof}
  We first use the algorithm from Lemma~\ref{thm:two-thirds} to derive
  instances $X' = (B'_x,\calS'_X)$ and $Y' = (B'_x,\calS'_X)$ such
  that $\#\xthreec(X) = \#\xthreec(X'')$, $\#\xthreec(Y) =
  \#\xthreec(Y'')$, $\frac{\frac{1}{3}\|B'_x\|}{\|\calS\|} =
  \frac{2}{3}$, and $\frac{\frac{1}{3}\|B'_x\|}{\|\calS\|} =
  \frac{2}{3}$. Without the loss of generality we can assume that
  $\|B'_x\| \leq \|B'_y\|$. We set $Y'' = Y'$ and derive $X''$ via
  repeatedly applying transformation $h'$ from
  Fact~\ref{fact:x3c-transform} to $X'$, until the condition of the
  theorem is met.~\qed
\end{proof}

In the next lemma and theorem we prove the $\pp$-completeness of
$\powercompare_\shapley$.

\begin{lemma}\label{thm:compare-via-ss}
Let $f$ and $g$ be two arbitrary $\sharpp$ functions. There exists
a polynomial-time computable function $\cmp(x,y)$ such that
$(\forall x,y \in \Sigma^*)[f(x) > g(y) \iff \cmp(x,y) \in \powercompare_\shapley]$.
\end{lemma}
\begin{proof}
  Let $f$ and $g$ be as in the lemma and let $x$ and $y$ be two
  arbitrary strings. %
  Since both $f$ and $g$ are in $\sharpp$ and $\#\xthreec$ is
  $\sharpp$-parsimonious-complete, there exist functions $\varphi_f$
  and $\varphi_g$ that compute parsimonious reductions from $f$ to
  $\#\xthreec$ and from $g$ to $\#\xthreec$, respectively.\footnote{We
    assume that neither $\varphi_f$ nor $\varphi_g$ ever outputs a
    malformed instance of $\xthreec$. This property is easy to
    enforce via the following modification:
    Whenever either $\varphi_f$ or $\varphi_g$ is about to output a malformed instance, 
    replace it with a fixed, correct one that has no solutions.}

  Let $(B_x,{\cal S}_x) = \varphi_f(x)$ and $(B_y,{\cal S}_y) =
  \varphi_g(y)$.  Via Lemma~\ref{thm:x3c-transform} (and through a slight
  abuse of notation) we ensure that
  $\|B_x\| = \|B_y\| = 3k$ and that $\|{\cal S}_x\| = \|{\cal S}_y\| =
  r$, where $r$ and $k$ are two nonnegative integers.  Let $\varphi$
  be the reduction function from the proof of
  Theorem~\ref{thm:many-one}. (Note that in the proof of
  Theorem~\ref{thm:many-one} we restricted $\varphi$ to work only on
  instances of $\xthreec$ that fulfill a special requirement. For the
  purpose of this proof we disregard this requirement.)

  We now describe our function $\cmp$.  Given an instance $(B_x,
  {\cal S}_x)$ we compute $G_x = \varphi(X)$ and $G_y = \varphi(Y)$.
  We define $\cmp(x,y)$ to output $(G_x, G_y, 1)$. 
  Via the properties of $\varphi$ discussed in the proof of
  Theorem~\ref{thm:many-one}, it holds that
  \begin{eqnarray*}
    \rawshapley(G_x,1) &=& (r - k)!k!\cdot \#\xthreec( B_x, {\cal
      S}_x) = (r - k)!k!f(x)\mbox{, and} \\
    \rawshapley(G_y,1) &=& (r - k)!k!\cdot \#\xthreec( B_y, {\cal
      S}_y) = (r - k)!k!g(y).
  \end{eqnarray*}
  Thus $f(x) > f(y)$ if and only if $\shapley(G_x,1) >
  \shapley(G_y,1)$, and so it is clear that the function $\cmp$ does what the
  theorem claims. Naturally, $\cmp$ can be computed in polynomial
  time.~\qed
\end{proof}

\begin{theorem}\label{thm:shapley}
  $\powercompare_\shapley$ is $\pp$-complete.
\end{theorem}
\begin{proof}
  Via Lemma~\ref{thm:pp-member} it is easy to see that
  $\powercompare_\shapley$ is in $\pp$. 
  Let $h$ be some $\sharpp$-parsimonious-complete function.
  $\pp$-hardness of $\powercompare_\shapley$ follows via a reduction
  from $\pp$-complete problem $\compare_h$ (see
  Lemma~\ref{thm:pp-com}). As a reduction we can use, e.g.,
  the function $\cmph$ from Lemma~\ref{thm:compare-via-ss}.
  This completes the  proof.~\qed
\end{proof}

\section{Conclusions and Open Problems}

We have shown that the problem of deciding in which of the two given
voting games our designated player has a higher power index value is
$\pp$-complete for both the Banzhaf and the Shapley-Shubik power
indices.  For the case of Banzhaf, we have used the fact that the raw
Banzhaf power index is $\sharpp$-parsimonious-complete. For the case
of Shapley-Shubik, we have shown that the raw Shapley-Shubik power
index is $\sharpp$-many-one-complete but not
$\sharpp$-parsimonious-complete. Nonetheless, using the index's
properties we were able to show the $\pp$-completeness of
$\powercompare_\shapley$.  We believe that these results are
interesting and practically important. Below we mention one particular
application. %

In the context of multiagent systems, the Shapley-Shubik power
index is often used to distribute players' payoffs, i.e.,
each player's payoff is proportional to his or her power index value.
Recently Elkind~\cite{elk:perscomm:joining} 
asked
about the exact
complexity of the following problem: Given a weighted voting game $G =
(w_1,w_2,\ldots, w_n;q)$, is it profitable for players $1$ and $2$ to
join? That is, if $G' = (w_1+w_2, w_3, \ldots, w_n; q)$, is it the
case that
$\shapley(G',1) > \shapley(G,1) + \shapley(G,2)$. Using
Lemma~\ref{thm:compare-via-ss} and the fact that $\sharpp$ is closed
under addition we can easily show that this problem reduces to
$\powercompare_\shapley$ and thus is in $\pp$. We believe that
Elkind's problem is, in fact, $\pp$-complete and that the techniques
presented in this paper will lead to the proof of this fact.
However,
at this point the exact complexity of the problem remains 
open.
\paragraph*{Acknowledgments}
We thank J{\"o}rg Rothe and Edith Elkind
for helpful discussions on the topic of weighted voting games and
J{\"o}rg Rothe for hosting a visit during which this work was done in
part.

\bibliography{grypiotr2006}

\end{document}